\documentclass[12pt,a4paper]{article}
\usepackage{amsmath,amssymb,amsthm,mathptmx,fullpage,graphicx,ifpdf,url}
\usepackage[left,pagewise]{lineno}
\DeclareMathOperator{\KS}{\mathit{C}}
\newcommand{\cnd}{\,|\,}
\newtheorem{thm}{Theorem}
\newtheorem{prop}[thm]{Proposition}
\newtheorem{lem}{Lemma}
\begin{document}

\title{Topological arguments for\\ Kolmogorov complexity}
\author{Andrei Romashchenko and Alexander Shen\thanks{Both authors work at the LIRMM, CNRS \& University of Montpellier 2 and are on leave from IITP RAS, Moscow. Supported in part by RFBR 14-01-93107 grant.}}
\date{}
\maketitle

\begin{abstract}

We present several applications of simple topological arguments (such as non-con\-tra\-cti\-bi\-li\-ty of a sphere and similar results) to Kolmogorov complexity. It turns out that discrete versions of these results can be used to prove the existence of strings with prescribed complexity with $O(1)$-precision (instead of usual $O(\log n)$-precision).

In particular, we improve an earlier result of M.~Vyugin and show that for every $n$ and for every string $x$ of complexity at least $n+O(\log n)$ there exists a string $y$ such that both $\KS(x\cnd y)$ and $\KS(y\cnd x)$ are equal to $n+O(1)$.
We also show that for a given tuple of strings $x_i$ (assuming they are almost independent) there exists another string $y$
such that the condition $y$ makes the complexities of all $x_i$ twice smaller with $O(1)$-precision.

The extended abstract of this paper was published in~\cite{dcm}.
\end{abstract}

\section{Introduction}

In this paper we discuss several applications of topological arguments in algorithmic information theory. The main notion of algorithmic information theory, the Kolmogorov (or algorithmic) complexity has a very simple definition\footnote{The plain Kolmogorov complexity $\KS(x)$ of a bit string $x$ is defined as the minimal length of a program that outputs $x$; this definition depends on the choice of a programming language, and we fix some language that makes the complexity minimal. The conditional complexity $\KS(x\cnd y)$ is defined as the minimal length of a program that transforms $y$ to $x$. The sum $\KS(y)+\KS(x\cnd y)$ is equal to the complexity $\KS(x,y)$ of (the encoding of) the pair $(x,y)$ with logarithmic precision. The mutual information between $x$ and $y$ is defined as $\KS(x)-\KS(x\cnd y)$, or $\KS(y)-\KS(y\cnd x)$, or as $\KS(x)+\KS(y)-\KS(x,y)$; all three quantities coincide with each other with logarithmic precision. Strings $x$ and $y$ are considered as  ``independent'' if $I(x:y)$ is negligible (of course, one has to specify the exact bound when formulating results about independent strings).  The information distance between strings $x$ and $y$ is defined as $\KS(x\cnd y)+\KS(y\cnd x)$; it satisfies the triangle inequality with logarithmic precision. 

We assume that the reader is familiar with all these notions. An introduction to Kolmgorov complexity can be found, e.g., in~\cite{gacsnotes,uppsalanotes,usv}; see also an extensive textbook on this subject written by Li and Vit\'anyi~\cite{li-vitanyi}.}, but even some basic questions about it turn out to be very difficult. For example, we still do not know which linear inequalities are true for the Kolmogorov complexities of tuples of strings, though we know that they are the same as the inequalities for Shannon entropies and also for the sizes of subgroups of a finite group and their intersections (see \cite[Chapter 10]{usv}).

Things could become even more complicated when we switch from universal statements (e.g., inequalities) to universal-existential statements. An example of such a statement: \emph{for every integer $n$ and for every $x$ of complexity at least $2n+O(1)$ there exists a string $y$ such that both $\KS(x\cnd y)$ and $\KS(y\cnd x)$ are equal to $n$ with $O(1)$-precision}. This result was proved by M.~Vyugin~\cite{mvyugin} using game technique developed by An.~Muchnik (see \cite{mmsv,msv} for some other examples of game arguments). Several other techniques to prove the existence of strings with given complexity properties were developed: one may use probabilistic method, and also some combinatorial techniques, in particular, on-line matching and expander-like constructions, see~\cite[Chapters 11, 12, 14]{usv} for examples. Still some basic questions remain open. For example, we do not know whether for every strings $a_1,\ldots,a_n$ there exists some string $b$ that, being used as a condition, makes the complexity of each $a_i$ twice smaller, i.e., $\KS(a_i\cnd b)=0.5\KS(a_i)+O(1)$ for every $i$.

In this paper we suggest one more technique which is quite different from the tools used before. It is applicable (as it seems) only in some rather special situations; however, when applicable, it provides results with maximal possible $O(1)$-precision. This approach is based on simple topological arguments.  To explain what is meant by ``topological arguments'', let us start with a very simple example. Let us consider a string $x$ that has Kolmogorov complexity $n$ (we consider plain complexity $\KS(x)$, but this does not matter for now). We want to prove that there exists a string $y$ such that $\KS(x\cnd y)\approx n/2$. Let us start with $y=x$, when $\KS(x\cnd y)\approx 0$. Then, we delete bits of the string $y$ (say, at the end) one by one until we get $y=\Lambda$ (the empty string) and $\KS(x\cnd y)\approx n$. During the process, each time when one  last bit of $y$ is deleted, the conditional complexity $\KS(x\cnd y)$ changes by at most $O(1)$, so it cannot cross the threshold $n/2$ without visiting $O(1)$-neighborhood of $n/2$.

Topological arguments of this type can be used in two (and more) dimensions, though they become less trivial. In what follows we provide several examples of this type.

\section{Vyugin's result and its extensions}\label{sec:vyugin-extension}

As we have said, M.~Vyugin~\cite{mvyugin} showed that for every $n$ and for every string $x$ with $\KS(x)\ge 2n+O(1)$ there exists a string $y$ such that both conditional complexities $\KS(x\cnd y)$ and $\KS(y\cnd x)$ are equal to $n+O(1)$. One may say informally that $y$ is ``$n$ bits apart from $x$ in both directions''. 

This result was proved in~\cite{mvyugin} using a rather ingenious game argument. As we shall see in this paper, the condition $\KS(x)\ge 2n$ is much stronger than necessary; it is enough to assume only that $\KS(x)\ge n+c\log n$ for some $c$. We present a  proof of this theorem based on a simple topological argument. This topological proof can be applied unless $\KS(x)$ is very large  (Vyugin's game argument still seems to be necessary if $\KS(x)$ is really big compared to $n$.)

Similar reasoning also allows us to construct $y$ such that both complexities $\KS(x\cnd y)$ and $\KS(y\cnd x)$ have prescribed values with $O(1)$-precision, even if those two values are different. (This question was discussed in Vyugin's paper~\cite{mvyugin}, but no positive result of this type was given there except for the already mentioned case $m=n+O(1)$.) Again we need some restrictions that guarantee that prescribed values are not unreasonable large or small. Here is the exact statement.

\begin{thm}\label{thm:distance}
   Let $P$ be some polynomial. There exists a constant $c$ such that for every string $x$ and for all integers $m,n$ such that
   \begin{itemize}
     \item $n+c\log n \le \KS(x) \le P(n)$,
     \item $c\log n \le m \le P(n)$,
   \end{itemize}
there exists a string $y$ such that $|\KS(x\cnd y)-n|\le c$ and $|\KS(y\cnd x)-m|\le c$.
\end{thm}
\noindent
\emph{Remark:} Let us comment on the assumptions that appear in this statement. We want to obtain $|\KS(x\cnd y)-n|\le c$. Since $\KS(x\cnd y)$ does not exceed $\KS(x)$, we need to assume that $\KS(x)$ is at least~$n$. The assumption in the theorem is a bit stronger: we assume a ``safety margin'' of logarithmic size and require $\KS(x)\ge n+c\log n $. We also require $m$ and $\KS(x)$ to be polynomially bounded.

\begin{proof}
Let $p$ be the shortest description of $x$; it is a string of length $\KS(x)$. Consider also an incompressible string $z$ of length slightly greater than $m$, e.g., let $|z|$ and $\KS(z)$ be equal to $2m+O(1)$. Moreover, we take $z$ independent from $p$, so $\KS(z\cnd p)$ is also $2m+O(1)$. (A randomly chosen $z$ has these properties with probability close to $1$, i.e., most strings of that length have the required properties.)

The string $y$ is then constructed as (the encoding of) a pair $(p', z')$ where $p'$ and $z'$ are prefixes of $p$ and $z$ respectively; it remains to decide how long $p'$ and $z'$ should be. In other words, we have two parameters, $|p'|$ and $|z'|$, and the space of the parameters is a rectangle of width $\KS(x)$ and height $2m$. Each point $(|p'|,|z'|)$ in this rectangle determines $y=(p',z')$. So we can map each point  $(|p'|,|z'|)$  to the pair of integers $(\KS(x\cnd y),\KS(y\cnd x))$. We need to show that some point is mapped to a point that is $O(1)$-close to $(n,m)$.

\begin{figure}[ht]
\begin{center}
\ifpdf
 \includegraphics[width=0.9\textwidth]{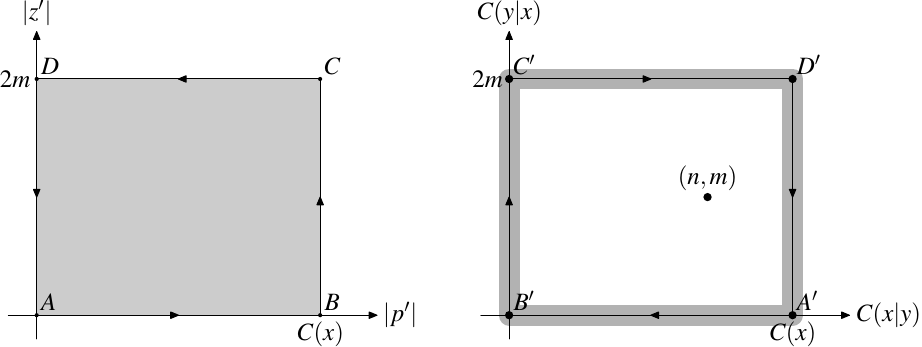}
\else
 \includegraphics[width=0.9\textwidth]{topology-1.mps}
\fi
\end{center}
\caption{Each pair in the left rectangle determines $y=(p',z')$ and is mapped into a pair $\KS(x\cnd y),\KS(y\cnd x)$ in the right one.}\label{fig1}
\end{figure}

To show this, we note that the mapping is ``continuous'' in the sense that neighbor points in the domain (on the left, Fig.~\ref{fig1}) are mapped into points at distance $O(1)$ (on the right, Fig.~\ref{fig1}).\footnote{A better name would be \emph{Lipschitz continuity}:  if the distance between the images of neighbor grid points is bounded by $c$, then the distance between the images of arbitrary two points is at most $c$ times bigger than the distance between the points itself (the distance is measured in $l_1$-sense)}
Indeed, $\KS(u\cnd v)$ changes only by $O(1)$ if $u$ or $v$ is changed by adding or deleting the last bit. Consider a path $A$--$B$--$C$--$D$--$A$ that goes counterclockwise around the rectangle on the left; as we shall see, the image path on the right will go clockwise (with logarithmic precision) around the rectangle and makes one turn around the point $(n,m)$. Then we can continuously transform the path on the left into one point (since the rectangle is simply connected); if its image on the right never comes close to $(n,m)$, we get a contradiction: the number of turns around $(n,m)$ cannot change if the image does not come close to that point.

This type of arguments is discussed later in Section~\ref{sec:top}, but for now an informal explanation could be sufficient. Imagine a fence on the left picture whose posts are the integer points on the path $A$--$B$--$C$--$D$--$A$, and the ``image fence'' on the right whose posts are images of the posts on the left. Both fences are polygonal closed curves whose vertices are posts; the distance between neighbor posts is $1$ on the left and $O(1)$ on the right. When we move one post on the left to the neighbour grid point, the fence on the left changes slightly, and the image fence also changes slightly (the image of the moving post changes its position by $O(1)$). In several steps we can shrink the fence on the left completely (moving all post to the same point); then the image fence also shrinks. But it is not possible if the posts on the right never come close to some point $(n,m)$ that initially was inside the fence.

It remains to look closely at the path around the rectangle and its image; we need to check that indeed it makes a turn around $(n,m)$. Note that our assumptions guarantee that $\log \KS(x)=\Theta(\log n)$, so we write just $O(\log)$ having in mind $O(\log n)$ or $O(\log\KS(x))$; note also that $O(\log m)\le O(\log)$.

\begin{itemize}
\item Point $A$: here $y=(\Lambda,\Lambda)$, so $\KS(y\cnd x)=O(1)$ and $\KS(x\cnd y)=\KS(x)+O(1)$. Thus, the image is $A'=(\KS(x),0)$ with $O(1)$-precision.
\item Edge $A$--$B$: here $y=(p',\Lambda)$. Then $\KS(y\cnd x)=O(\log)$ since $\KS(p\cnd x)=O(\log\KS(x))$ [the conditional complexity of a shortest description of $x$ given $x$ is $O(\log\KS(x))$], and the length of $p'$ can be described by $O(\log\KS(x))$ bits, too. And $\KS(x\cnd y)$ is somewhere between $0$ and $\KS(x)+O(1)$. So the image of $A$--$B$ is in the $O(\log)$-neighborhood of $A'$--$B'$.
\item Point $B$: here $y=(p,\Lambda)$, so $\KS(y\cnd x)=O(\log\KS(x))$ and $\KS(x\cnd y)=O(1)$; the image is $B'=(0,0)$ with $O(\log)$-precision; the edge $A$--$B$ is mapped into a path along $A'B'$ going from $A'$ to $B'$ with $O(\log)$-precision.
\item Edge $B$--$C$: here $y=(p,z')$, so $\KS(y\cnd x)$ is somewhere between $0$ and $2m+O(\log)$ [recall that the length of $z'$ is between $0$ and $2m$ and can be described by $O(\log m)=O(\log)$ bits; $\KS(p\cnd x)$ is also $O(\log)$]. On the other hand, $\KS(x\cnd y)=O(1)$, since $p$ determines~$x$.
\item Point $C$: here $y=(p,z)$, so $\KS(x\cnd y)=O(1)$ and $\KS(y\cnd x)=O(\log)+\KS(z\cnd x)=2m+O(\log)$. So the image is $C'=(0,2m)$ with $O(\log)$-precision.
\item Edge $C$--$D$: here $y=(p',z)$, so $\KS(y\cnd x)=2m+O(\log)$ and $\KS(x\cnd y)$ is between $0$ and $\KS(x)+O(\log)$.
\item Point $D$: here $y=(\Lambda,z)$, so $\KS(y\cnd x)=\KS(z\cnd x)+O(1)=2m+O(1)$ and $\KS(x\cnd y)=\KS(x\cnd z)=\KS(x)+O(\log)$ since $x$ and $z$ have only logarithmic mutual information. So the image of $D$ is $D'=(\KS(x),2m)$ with $O(\log)$-precision.
\item Edge $D$--$A$: here $y=(\Lambda,z')$, so $\KS(y\cnd x)$ is between $0$ and $2m+O(\log)$, and $\KS(x\cnd y)$ is $\KS(x)+O(\log)$ (note that $z'$ can have only $O(\log)$ bits of additional information compared to $z$).
\end{itemize}

This analysis shows that the path on the right follows the trajectory $A'$--$B'$--$C'$--$D'$--$A'$ with $O(\log)$-precision and therefore turns around the point $(n,m)$ if this point is $O(\log)$-far from the boundary of the rectangle, and this is exactly what our assumption guarantees.  Theorem~\ref{thm:distance} is proven.
\end{proof}

\section{Topological digression}\label{sec:top}

Let us look more closely at the topological arguments we used. First let us recall the following basic topological result: \emph{the circle is not contractible}. This means that the identity mapping $\textrm{id}\colon S^1\to S^1$ of a circle cannot be extended to a continuous mapping of a disc $D^2\to S^1$. Another version of essentially the same result: the identity mapping of a circle $S^1$ is not homotopic (cannot be continuously transformed into) a constant mapping.

How these results are proved usually? The most intuitive argument uses the winding number, an integer representing the total number of times that curve travels counterclockwise around the point. For an identity mapping the winding number is $1$, and  for the constant mapping it is zero; therefore, these two curves are not homotopic (the continuous change in the curve should change the winding number continuously; since it is an integer, it does not change).

We can consider any point $z$ inside the circle and consider the winding number with respect to $z$. The same argument then can be applied not only to mappings $D^2\to S^1$ but also to mappings $D^2\to D^2$ that do not cover $z$, since the winding number is well defined for those curves. In this way we get the following statement: \emph{every mapping $F\colon D^2\to D^2$ that extends the identity mapping of the circle, covers the entire disk}.  In fact,  we do not need that $F(z)$ is exactly $z$ for $z\in S^1$; if $F(z)$ is $\varepsilon$-close to $z$ for $z\in S^1$, then (for the same reasons) the image of $F$ covers the entire $D^2$ except (maybe) the $\varepsilon$-zone near the boundary.

In Section~\ref{sec:vyugin-extension} we used the discrete version of this argument (with a rectangle grid instead of a circle). One can define a winding number of a discrete path (sequence of points, or vertices of a polygonal line) assuming that vertices do not come close to the center point; this number is an integer and does not change if the vertices of the path are moved slightly (the change of the position of each point should be a small fraction of its distance to the center point).

More generally, many topological existence theorems have natural finite versions. For example, the famous Brouwer fixed-point theorem says that every continuous mapping $F\colon I^n\to I^n$ of a $n$-dimensional cube into itself has a fixed point. (One can speak about a disk or a simplex instead of a cube; all these bodies are homeomorphic.) The finite version replaces the cube by a rectangular grid $G^n=\{0,1,\ldots,N\}^n$; if a mapping $F\colon G^n\to G^n$ is $c$-Lipschitz (increases distances at most $c$ times), then there exists a point $x\in G^n$ such that $F(x)$ is $O(c)$-close to $x$.

Note that one can reduce this discrete version of the fixed-point theorem to the standard continuous version by extending the mapping from the grid to a continuous mapping of the entire cube in $\mathbb{R}^n$ to $\mathbb{R}^n$; it can be done in a piecewise-affine way, so that the image of each unit cube belongs to the convex hull of the images of its vertices. Then we take a fixed point of the continuous mapping and use $c$-Lipschitz property to show that neighbor grid point changes its position only by $O(1)$.

However, this reduction (from a discrete version of the fixed-point theorem to a continuous one) is a bit silly since the standard (elementary) proof of Brouwer's fixed point theorem reduces it to the discrete version, and the discrete statement is proved using Sperner's lemma. Anyway, we will use in the sequel several discrete version of elementary topological results and hope that the reader can easily reconstruct their proofs.

\section{Decreasing complexity by adding a condition}\label{sec:decreasing}

Let $a$ and $b$ be two strings. They have some complexities $\KS(a)$ and $\KS(b)$. If a third string $y$ is given, we can consider the conditional complexities $\KS(a\cnd y)$ and $\KS(b\cnd y)$ which are (in general) smaller that $\KS(a)$ and $\KS(b)$. Now the question: can we describe the pairs $(\KS(a\cnd y),\KS(b\cnd y))$ that can be obtained by choosing an appropriate value of~$y$? We answer this question for the case when $a$ and $b$ have negligible mutual information. In this case the answer is simple: we can get an arbitrary pair  $(\alpha,\beta)$ such that $0\le\alpha\le \KS(a)$ and $0\le\beta\le\KS(b)$ and $\alpha,\beta$ are not too close to the endpoints of the corresponding intervals (the distance is big compared to the logarithms of complexities and to the mutual information).

\begin{thm}\label{thm2}
 For some constant $c$ the following statement holds: for every two strings $a,b$ and for every integers $\alpha,\beta$ such that
 \begin{itemize}
 \item $\alpha,\beta \ge c(\log\KS(a)+\log\KS(b)+I(a:b))$;
 \item $\alpha \le \KS(a) - c(\log\KS(a)+\log\KS(b)+I(a:b))$;
  \item $\beta \le \KS(b) - c(\log\KS(a)+\log\KS(b)+I(a:b))$,
 \end{itemize}
 there exists a string $y$ such that $|\KS(a\cnd y)-\alpha|\le c$ and $|\KS(b\cnd y)-\beta|\le c$.
\end{thm}

\begin{proof}
First of all, note that this statement is easy to prove if instead of $O(1)$-precision we are satisfied with $O(\log\KS(a)+\log\KS(b)+I(a:b))$-precision. Indeed, consider the shortest descriptions $p$ and $q$ for $a$ and $b$ and then let $y=(p',q')$ where $p'$ is $p$ without $\alpha$ last bits, $q'$ is $q$ without $\beta$ last bits.  The strings $p$ and $q$ are incompressible, so $\KS(p\cnd p')=\alpha$ and $\KS(q\cnd q')=\beta$ with  $O(\log\KS(a))$ and $O(\log\KS(b))$ precision (respectively). Note also that the information distance between $a$ and $p$ is $O(\log\KS(a))$,  the information distance between $b$ and $q$ is $O(\log\KS(b))$, so $\KS(a\cnd p')$ and $\KS(b\cnd q')$ are equal to $\alpha$ an $\beta$ with the same precision. What happens if we add $q'$ (respectively $p'$) to the condition and consider $\KS(a\cnd (p',q'))$ and $\KS(b\cnd(p',q'))$? Since $q'$ has logarithmic complexity given $q$ and therefore logarithmic complexity given $b$, adding $q'$ to the condition changes the complexity of $a$ at most by $O(\log\KS(b)+I(a:b))$; similar statement is true for $p'$, so we get the desired result.

To get $O(1)$-precision, we need to combine the simple construction above with a topological argument similar to the proof of Theorem~\ref{thm:distance}. Consider the shortest descriptions $p$ and $q$ for $a$ and $b$. Then $|p|=\KS(a)$ and $|q|=\KS(b)$. For every pair $(u,v)$ of integers such that $0\le u\le |p|$ and $0\le v\le |q|$ define $y(u,v)$ as follows:
    $$
 y(u,v)=(\text{$p$ without $u$ last bits},\ \text{$q$ without $v$ last bits}).
    $$
 As we have discussed, $\KS(a\cnd y(u,v))$ and $\KS(b\cnd y(u,v))$ are close to $u$ and $v$ respectively; the distance is $O(\log\KS(a)+\log\KS(b)+I(a:b))$.

In other terms, we consider the mapping
   $$
(u,v)\mapsto (\KS(a\cnd y(u,v)),\KS(b\cnd y(u,v)).
   $$
It is defined on the rectangle $[0,\KS(x)]\times[0,\KS(y)]$ and is close to the identity mapping; for each point $(u,v)$ the distance between this point  and its image is at most $O(\log\KS(a)+\log\KS(b)+I(a:b))$. This mapping is also continuous in the sense explained above. Now the topological argument (similar to the one used in Section~\ref{sec:vyugin-extension}) can be used to show that the image $O(1)$-covers the rectangle except for $O(\log\KS(a)+\log\KS(b)+I(a:b))$-neighborhood of its boundary. Indeed, the image of the boundary of the rectangle makes a turn around every point inside the rectangle (and not too close to the boundary), and therefore all these points are $O(1)$-close to the image of this mapping.
\end{proof}

\emph{Remark:} This argument can be generalized easily to three (or more) dimensions. For example, let us consider three strings $a,b,c$ that are almost independent. In this case we get a mapping of a three-dimensional box to itself which is ``continuous'' and is close to identity. Then a topological argument (based on the fact that identity mapping of the two-dimensional sphere $S^2$ is not homotopic to the constant mapping) shows that the image of this mapping covers (with $O(1)$-precision) the interior of the box.

\section{Combination with Muchnik's technique}

For the case when $a$ and $b$ are dependent, the result of Theorem~\ref{thm2} looks rather weak. We can extent the area of pairs $(\alpha,\beta)$ that can be covered, if we combine the topological technique with an argument based on Muchnik's result on  conditional descriptions~\cite{muchnik}. Let us recall first Muchnik's result.

\begin{prop}[Muchnik]\label{muchnik}
\textup{\textbf{(a)}}~Let $x$ and $y$ be arbitrary strings of length at most $n$. Then there exists
a string $p$ of length $\KS(x\cnd y)$ such that
\begin{itemize}
\item $\KS (p\cnd x) = O(\log n)$ and
\item $\KS (x\cnd p, y) = O(\log n)$.
\end{itemize}
\textup{\textbf{(b)}}~Consider a string $x$ and a family of strings $y_1,\ldots,y_m$; assume that the number of strings and their lengths are bounded by some $n$. Then there exists a string $p$ such that
\begin{itemize}
\item $\KS (p\cnd x) = O(\log n)$ and
\item for every $j=1,2,\ldots,m$ the complexity $\KS(x\cnd y_j,p_j)$ is $O(\log n)$, where $p_j$ is a prefix of $p$ having length $\KS(x\cnd y_j)$.
\end{itemize}

\noindent
As usual, the constants in $O(\log n)$-notation do not depend on $n$.
\end{prop}

\noindent
\emph{Remark:} We cannot get rid of logarithmic terms in this theorem, i.e., we cannot obtain $\KS (p\cnd x) = O(1)$ instead of $\KS (p\cnd x) = O(\log n)$, see \cite{vereshchagin}. However, in what follows we show that a combination of Muchnik's theorem with a topological argument can result in a statement with $O(1)$-precision.

\smallskip

The first part of the proposition says that the ``information difference'' between $x$ and $y$ (``information that is present in $x$ but is missing in $y$'') can be ``materialized'' as some string $p$. The second part of the proposition says that for several strings $y_1,\ldots,y_m$ one can choose the representatives of the information differences $x\setminus y_i$ to be prefixes of each other.

The string $p$ provided by the first part of this Proposition is incompressible (with logarithmic precision).  Since $p$ and $y$ together are enough to reconstruct the pair $(x,y)$ and the sum of complexities of $p$ and $y$ is $O(\log n)$-close to the complexity of this pair (due to Kolmogorov--Levin theorem about the complexity of a pair), the strings $p$ and $y$ are independent with logarithmic precision, i.e., $I(p:y)=O(\log n)$. In the second part the prefixes of $p$ (that are actually used) are also incompressible, and  $y_j$ has only $O(\log n)$ mutual information with the prefix of $p$ that has length $\KS(x\cnd y_j)$.

Now we use this result as a tool to improve Theorem~\ref{thm2} and get the following

\begin{thm}\label{thm3}
 For some constant $c$ the following statement holds: for every two strings $a,b$ of complexity at most $n$ and for every integers $\alpha,\beta$ such that
 \begin{itemize}
 \item $\alpha \le \KS(a) - c\log n$,
  \item $\beta \le \KS(b) - c\log n$,
   \item $-\KS(a\cnd b)+c\log n\le\beta-\alpha \le \KS(b\cnd a) - c\log n$,
 \end{itemize}
 there exists a string $y$ such that $|\KS(a\cnd y)-\alpha|\le c$ and $|\KS(b\cnd y)-\beta|\le c$.
\end{thm}

\emph{Remark}. The conditions for $\alpha$ and $\beta$ say that the point $(\alpha,\beta)$ is inside the hexagon shown by Figure~\ref{fig-hexagon} below.

\begin{proof}

In the argument above we considered the shortest descriptions $p$ and $q$ for $a$ and $b$, i.e., $p$ and $q$ were incompressible strings that contain the same information as $a$ and $b$ respectively. Now we need to choose these two strings $p$ and $q$ in a special way: their prefixes should be independent as much as possible and become dependent only if the total length of the two prefixes exceeds $\KS(a,b)$. This can be done using Muchnik's result, and then we can use $p$ and $q$ in the same way as before.

We start by saying precisely what are the required properties of $p$ and $q$.

\begin{lem}\label{lemma}
For all strings $a$ and $b$ of complexity at most $n$ there exist strings $p$ and $q$ such that
\begin{itemize}
\item $|p|=\KS(a)$;
\item $|q|=\KS(b)$;
\item the information distance between $p$ and $a$ is $O(\log n)$;
\item the information distance between $q$ and $b$ is $O(\log n)$;
\item if $l\le |p|$, $m\le |q|$ and $l+m=\KS(a,b)$, then $I(p_l : q_m ) = O(\log n)$ and
the information distance between $(a,b)$ and $(p_l,q_m)$ is $O(\log n)$.
\end{itemize}
where $p_l$, $q_m$ denote  prefixes of strings $p$ and $q$ of length $l$
and $m$ respectively, and information distance between $u$ and $v$ is defined as $\KS(u\cnd v)+\KS(v\cnd u)$.
\end{lem}

One can say that $p$ and $q$ present the same information as $a$ and $b$ in such a way that prefixes of $p$ and prefixes of $q$ are independent as much as possible --- up to the point where they become too long to be independent, since the total amount of information is only $\KS(a,b)$. In particular, if $l=\KS(a\cnd b)$, then the last requirement says that $p_l$ is independent with $b$ (have logarithmic mutual information); all shorter prefixes of $p$ are independent with $b$, too. One could note also that the information distance requirement is the consequence of independence: $p_l$ and $q_m$ are simple given $(a,b)$, and if $l+m=\KS(a,b)$, they together contains as much information as the pair $(a,b)$ itself.

A technical remark: we use the exact equalities for lengths, not $O(\log n)$-precision ones, as for the complexities. This does not really matter, since adding or deleting $O(\log n)$ bits changes all the complexities and information quantities by $O(\log n)$ only.

\begin{proof}[Proof of the lemma:]

Though $a$ and $b$ appear in the same way, our proof is non-symmetric: first we find one of the strings $p$ and $q$ (say, $p$) and then construct $q$ using $p$.

First, we apply the Proposition~\ref{muchnik} to get a string $p$ that starts with the ``difference'' $a\setminus b$ and then presents the rest of~$a$ . This means that
\begin{itemize}
\item $|p|=\KS(a)$;
\item $\KS(p\cnd a)=O(\log n)$;
\item $\KS(a\cnd p)=O(\log n)$;
\item $\KS(a\cnd b,p')=O(\log n)$, where $p'$ is the prefix of $p$ of length $\KS(a\cnd b)$.
\end{itemize}
For that we can apply the part (b) of this proposition to two strings: $b$ and the empty string. Another option is to apply the part (a) to get the prefix of $p$ and then add the conditional description of $a$ given this prefix.

As we have mentioned, the strings $p'$ and $b$ have almost no mutual information, i.e., their mutual information is $O(\log n)$. The same is therefore true for all shorter prefixes of $p$: they also have almost no mutual information with $b$.

Then, we apply Proposition~\ref{muchnik}, part (b), to construct a
description of $b$ using all prefixes of $p$ as $y_1,\ldots,y_m$. We get a string $q$ such that
\begin{itemize}
\item $|q|=\KS(b)$;
\item $\KS(q\cnd b)=O(\log n)$;
\item $\KS(b\cnd q)=O(\log n)$;
\item $\KS(b\cnd p_l, q_m)=O(\log n)$ for all $l\le |p|$ and $m\le |q|$ such that $l+m =\KS(a, b)$.
\end{itemize}
Here $p_l$ and $q_m$ are prefixes of $p$ and $q$ of lengths $l$ and $m$ respectively; to see why the last requirement is provided by Proposition~\ref{muchnik}, we note that $\KS(b\cnd p_l)$ does not exceed $m=\KS(a,b)-l$, because $\KS(b\cnd p)$ is equal to $\KS(b\cnd a)$ and only the missing bits of $p$ should be reconstructed:
     $$
\KS(b\cnd p_l)\le \KS(b\cnd p)+(|p|-l)=[\KS(b\cnd p)+|p|]-l=\KS(a,b)-l=m
     $$
(all the equalities and inequalities are understood with $O(\log n)$ precision). It is now easy to see that the mutual information $I(p_l:q_m)$ for all $l$ and $m$ such that $l+m=\KS(a,b)$ is $O(\log n)$, otherwise $p_l$ and $q_m$ would contain not enough information to reconstruct $a$ and $b$.
\end{proof}

Now we are ready to finish the proof of Theorem~\ref{thm3}, combining the argument used for Theorem~\ref{thm2} and Lemma~\ref{lemma}.
Let $p$ and $q$ be the strings provided by this Lemma.
For every pair $(l,m)$ of integers such that $0\le l\le |p|$ and $0\le m\le |q|$ we define
    $$
 y(l,m)=(\text{first $l$ bits of $p$},\ \text{first $m$ bits of $q$}).
    $$
Similarly to the proof of Theorem~\ref{thm2}, we consider the mapping
   $$
(l,m)\mapsto (\KS(a\cnd y(l,m)),\KS(b\cnd y(l,m)),
   $$
defined on the rectangle $[0,\KS(a)]\times[0,\KS(b)]$. (Recall that $p$ and $q$ are incompressible versions of $a$ and $b$ and therefore have lengths $\KS(a)$ and $\KS(b)$ respectively; we ignore the $O(\log n)$-difference for simplicity.)

 \begin{figure}[ht]
\begin{center}
\ifpdf
 \includegraphics[width=1.0\textwidth]{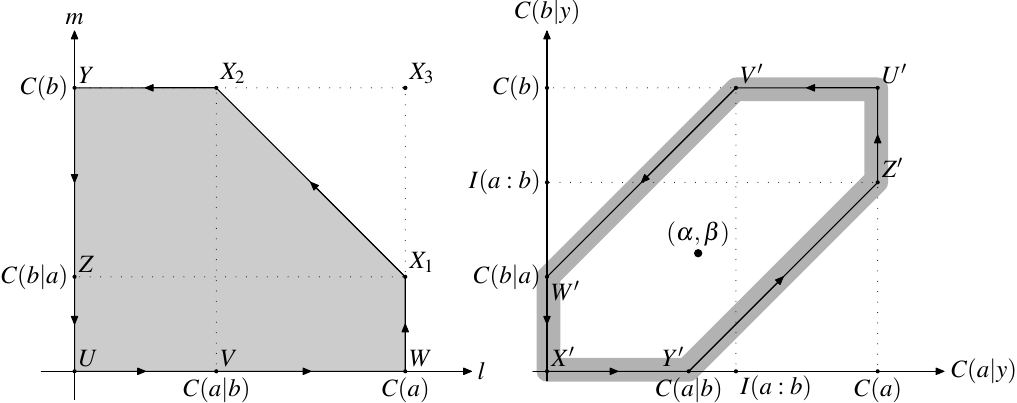}
\else
 \includegraphics[width=1.0\textwidth]{topology-2.mps}
\fi
\end{center}
\caption{Each pair of integers on the left determines $y=(a'_l,b'_m)$, which is mapped to a pair $(\KS(a\cnd y),\KS(b\cnd y))$ on the right.}\label{fig-hexagon}
\end{figure}

This mapping is ``continuous''  (neighbor points are mapped into points at distance $O(1)$). Consider the path $U$--$V$--$W$--$X_1$--$X_2$--$Y$--$Z$--$U$ that goes counterclockwise around the pentagon on Fig.~\ref{fig-hexagon} (left).
We will show that the image path  will (with logarithmic precision) go clockwise around the hexagon
$U'$--$V'$--$W'$--$X'$--$Y'$--$Z'$--$U'$ on Fig.~\ref{fig-hexagon} (right).  This path
makes a turn around the point $(\alpha,\beta)$. Hence, point $(\alpha,\beta)$ (more precisely,
some point in $O(1)$-neighborhood of  $(\alpha,\beta)$)
has a preimage $(l,m)$ in  the rectangle $[0,\KS(a)]\times[0,\KS(b)]$.

It remains to look closely at the path around the pentagon and its image.

\begin{itemize}

\item $m=0$ and $l=0\ldots \KS(a\cnd b)$: the image goes along $U'$--$V'$ with $O(\log n)$-precision  (recall that prefixes of $p$ up to length $\KS(a\cnd b)$ are independent with $b$ with the same precision; all the statements in this argument are understood in the same way);

\item $m=0$ and $l=\KS(a\cnd b)\ldots \KS(a)$: the image goes along $V'$--$W'$ with $O(\log n)$-precision; now the increase in $p_l$ means the same decrease in $\KS(b\cnd p_l)$. Indeed, for $m=\KS(a,b)-l$ we have
    $$\KS(b\cnd p_l)=\KS(b,p_l\cnd p_l)=\KS(a,b\cnd p_l)=|q_m|=m.$$
Indeed, since $l\ge \KS(a\cnd b)$, then $p_l$ and $b$ together are enough to reconstruct $a$, so the second equality is true. The third one holds because $p_l$ and $q_m$ are independent and $(a,b)$ is informationally close to $(p_l,q_m)$.

\item $l=\KS(a)$ and $m=0\ldots \KS(b\cnd a)$: the image goes along $W'$--$X'$ with $O(\log n)$-precision. Indeed, $y(l,m)$ now has full information about $a$; to compute $\KS(b\cnd y(l,m))$ we note that $y(k,l)$ is close to $(a,q_m)$ and $\KS(b\cnd a,q_m)=\KS(a,b\cnd a,q_m)=\KS(a,b)-\KS(a,q_m)$; recall that $q_m$ is simple given $b$ and independent with $a$.

\item $l=\KS(a)-\lambda$ and $m=\KS(b\cnd a)+\lambda$ for $\lambda=0\ldots I(a:b)$:
the image remains  in $O(\log n)$-neighborhood of $X'$. Indeed, we use prefixes $p_l$ and $q_m$ with $l+m=\KS(a,b)$, and these prefixes are enough to reconstruct the pair $(a,b)$. Note also that $\KS(a)+\KS(b\cnd a)$ and $\KS(b)+\KS(a\cnd b)$ are both equal to $\KS(a,b)$, so the slope of the line is $-1$.

The rest of the path is symmetric:

\item $m=\KS(b)$ and $l=\KS(a\cnd b)\ldots 0$: the image goes along $X'$--$Y'$ with $O(\log n)$-precision;
\item $l=0$ and $m=\KS(b)\ldots \KS(b\cnd a)$: the image goes along $Y'$--$Z'$ with $O(\log n)$-precision;
\item $l=0$ and $m=\KS(b\cnd a)\ldots 0$: the image goes along $Z'$--$U'$ with $O(\log n)$-precision;
\end{itemize}

Thus, the path on in the image follows the trajectory $U'$--$V'$--$W'$--$X'$--$Y'$--$Z'$--$U'$
with $O(\log n)$-precision. Therefore it turns around the point $(\alpha,\beta)$ if this point is $O(\log n)$-far from the boundary of the hexagon, and this is exactly what our assumption guarantees.
 \end{proof}

 \emph{Remark:}  Instead of the path $U$--$W$--$X_1$--$X_2$--$Y$--$U$
 we could take another path  $U$--$W$--$X_3$--$Y$--$U$.
 The shortcut $X_1$--$X_2$ is equivalent to the longer path $X_1$--$X_3$--$X_2$ since
all points of the triangle $X_1 X_2 X_3$ are mapped into
$O(\log n)$-neighborhood of $X'$ (if $l+m\ge \KS(a,b)$, then $y(l,m)$ contains
enough information to reconstruct both $a$ and $b$).

From Theorem~\ref{thm3} it follows in particular that  if $\KS(a\cnd b)$ and $\KS(b\cnd a)$ are \emph{not} logarithmically negligible, one can cut by factor $2$ the complexities of $a$ and $b$ by adding a condition. The problem for the case when $\KS(a\cnd b)$ and $\KS(b\cnd a)$ are both small is that the hexagon degenerates into a line, and we cannot apply our technique. In this case it is easy to cut the complexities by factor $2$ with logarithmic precision, just by doing this for one string and noting that the other one is close to the first one, but we do not know whether the $O(1)$-result is true. 

 \section{The $3$-dimensional case}

We have already briefly mentioned a $3$-dimensional version of Theorem~\ref{thm2} at the end of Section~\ref{sec:decreasing}. This result also can be improved using Muchnik's technique. In this section we explain this argument in some detail. Let us mention that Theorem~\ref{thm4} below looks not very impressive; we believe that a somewhat stronger version of this theorem can be proven with a similar technique (but now we cannot even formulate an exact statement as a conjecture).

\begin{thm}\label{thm4}
 For some constant $c$ the following statement holds: for every three strings $x,y,z$ of length at most $n$ and for every integers $\alpha,\beta,\gamma$ such that
 \begin{itemize}
 \item $I(x:y,z) + c\log n\le \alpha \le \KS(x) - c\log n$,
  \item $I(y:x,z) + c\log n\le\beta \le \KS(y) - c\log n$,
  \item $I(z:x,y) + c\log n\le\gamma \le \KS(z) - c\log n$,
 \end{itemize}
there exists a string $s$ such that $|\KS(x\cnd s)-\alpha|\le c$,  $|\KS(y\cnd s)-\beta|\le c$, and $|\KS(z\cnd s)-\gamma|\le c$.
\end{thm}
\noindent
\emph{Remark:}
The difference between Theorem~\ref{thm4} and the mentioned Remark at the end of Section~\ref{sec:decreasing} can be described as follows. In Section~\ref{sec:decreasing} we assumed that three strings $x,y,z$ are almost independent: if all the dependencies are bounded by some number $d$, then the result can be obtained for points that are $O(d)$-inside the box of size $\KS(x)\times\KS(y)\times\KS(z)$. Now we look at the situation more closely and note that the dependence  between $x,y,z$ do not make the general statement completely false but only shifts  the lower bounds for $\alpha,\beta,\gamma$. We prove this statement using Muchnik's technique.

\begin{proof}
First, we apply Muchnik's theorem on conditional descriptions and then improve the precision using a topological argument. Proposition~\ref{muchnik} provides a string $x'$ of size and complexity $\KS(x\cnd y,z)+O(\log n)$ that is simple given $x$, i.e., $\KS(x'\cnd x)=O(\log n)$, and independent with the pair $(y,z)$, i.e.,  $I(x':y,z)=O(\log n)$. This $x'$ is a ``difference'' $x\setminus (y,z)$. In a similar way we construct strings $y'=y\setminus(x,z)$ and $z'=z\setminus(x,y)$. It is easy to see that the triple $x',y',z'$ is independent (the complexity of the triple is equal to the sum of complexities) and the same is true for all prefixes (simple algorithmic transformations, like taking a prefix, can only decrease mutual information).

Now for each triple $(k,l,m)$ with
$$0\le k\le\KS(x\cnd y,z), \ 0\le l \le \KS(y\cnd x,z), \  0\le m\le \KS(z\cnd x,y)$$
we consider the prefixes $x'_k$, $y'_l$, $z'_m$ of strings $x$, $y$, and $z$ having lengths $k$, $l$, and $m$ respectively, the triple of strings
$$ s(k,l,m)=(x'_k, y'_l, z'_m) $$
and the triple of numbers
$$F(k,l,m)=\bigl(\KS(x)-\KS(x\cnd s(k,l,m)), \KS(y)-\KS(y\cnd s(k,l,m)), \KS(z)-\KS(z\cnd s(k,l,m))\bigr).$$

The construction of $x'$, $y'$ and $z'$ guarantees that  $F(k,l,m)$ is $O(\log n)$-close to $(k,l,m)$ for all triples $(k,l,m)$ from the domain of $F$ (described above). Let us check this property. The situation is symmetric, so we consider $x$ and have to show that
$$ \KS(x\cnd x'_k,y'_l,z'_m)=\KS(x)-k.$$
Intuitively it looks plausible, since $x'_k$ contains $k$ bits of information from $x$ while $y'_l$ and $z'_m$ are prefixes of strings $y'$ and $z'$ that are independent with $x$. The formal argument goes as follows. Using $x'_k$ as a condition, we decrease the complexity of $x$ by $k$ (since $x'_k$ is simple given $x$ and is incompressible). Adding $y'_l$ to the condition, we do not change the complexity, since $y'$ is independent with $x$ (and therefore with $x'_k$, since $x'_k$ is simple given $x$). And adding a condition that is independent with all other strings, we do not change the complexities. Finally, we have to use this observation once more: we add condition $z'_m$ that is independent with the pair $(x,y)$ and therefore also with $x'_k$ and $y'_l$.

Thus, the mapping $(k,l,m)\mapsto F(k,l,m)$
is defined on the  parallelepiped
 $$
 [0,\KS(x\cnd y,z)]\times[0,\KS(y\cnd x,z)]\times[0,\KS(z\cnd x,y)]
 $$
and is $O(\log n)$-close to the identity mapping. Also $F(k,l,m)$ is continuous: changing $k$, $l$ or $m$ by $1$ we change the value of the function by $O(1)$. Now we can apply a $3$-dimensional topological argument and conclude that
every point in the $O(\log n)$-interior of the parallelepiped is $O(1)$-close to some value $F(k,l,m)$.
Hence, for every triple $(\alpha,\beta,\gamma)$ satisfying the conditions of the theorem there exists a triple $(k,l,m)$ such that for $s=s(k,l,m)$ we have
$\KS(x\cnd s)=\alpha+O(1)$,  $\KS(y\cnd s)=\beta+O(1)$, and $\KS(z\cnd s)=\gamma+O(1)$.
\end{proof}

\begin{center} * \quad * \quad * \end{center}

\textbf{Acknowledgements}. The authors are grateful to Laurent Bienvenu who suggested to write down this simple argument,  Tarik Kaced, and all the colleagues in Escape/NAFIT/Kolmogorov seminar team. We thank the (anonymous) referee for a very detailed review that pointed out many inaccuracies and suggested several improvements.


\begin{thebibliography}{9}

\bibitem{vereshchagin}  Bruno Bauwens, Anton Makhlin, Nikolay Vereshchagin, Marius Zimand. Short Lists with Short Programs in Short Time. 
\emph{Proceedings of the 28th IEEE Conference on Computational Complexity} (2013), p.~98--108.

\bibitem{gacsnotes} Peter G\'acs, \emph{Lecture Notes on Descriptional Complexity and Randomness},\\
\texttt{http://www.cs.bu.edu/faculty/gacs/papers/ait-notes.pdf}

\bibitem{li-vitanyi} Ming Li and Paul Vit\'anyi, \emph{An Introduction to Kolmogorov Complexity and Its Applications.} 3rd edition, Springer Verlag, 2008. 

\bibitem{muchnik} Andrej~Muchnik, Conditional complexity and codes, \emph{Theoretical Computer Science},
\textbf{271}(1--2): 97--109 (2002).

\bibitem{mmsv}
Andrej Muchnik, Ilya Mezhirov, Alexander Shen, Nikolay Vereshchagin, \emph{Game interpretation of Kolmogorov complexity}, \url{arxiv.org/abs/1003.4712}

\bibitem{dcm}
Andrei Romashchenko and  Alexander Shen. Topological arguments for Kolmogorov complexity, \emph{Proceedings of the 18th international workshop on Cellular Automata and 3rd international symposium Journ{\'e}es Automates Cellulaires}, EPTCS, \textbf{90} (2012), p.~127--132.

\bibitem{uppsalanotes} Alexander~Shen, \emph{Algorithmic information theory and Kolmogorov complexity}, technical report TR2000-034, Uppsala University.\\ \texttt{http://www.it.uu.se/research/publications/reports/2000-034}

\bibitem{msv}
Alexander Shen, \emph{Game arguments in computability theory and algorithmic information theory}, \emph{Computability in Europe}, 2012, LNCS \textbf{7318} (2012), p.655--666. Extended version (with Andrej Muchnik and Mikhail Vyugin): \url{arxiv.org/abs/1204.0198}.

\bibitem{usv} Nikolay Vereshchagin, Vladimir Uspensky and  Alexander Shen, \emph{Kolmogorov complexity and algorithmic randomness}, Moscow, MCCME Publishers, 2012. 575 p. For the electronic version and draft English translation see \url{www.lirmm.fr/~ashen}

\bibitem{mvyugin}
Mikhail~Vyugin, Information distances and conditional complexities, \emph{Theoretical Computer Science}, \textbf{271}(1--2): 145--150 (2002).

\end{thebibliography}
\end{document}